\documentclass{article}
\usepackage[utf8]{inputenc}
\usepackage[usenames,dvipsnames]{xcolor}
\usepackage{algorithm,soul,algpseudocode,subcaption,bm,bbm}
\usepackage{amsmath}
\usepackage{leftidx}
\usepackage{amsthm}
\usepackage{amssymb}
\usepackage{mathtools}
\usepackage{enumerate}
\usepackage{multicol}
\usepackage{float}
\usepackage{graphicx}
\usepackage{tikz}
\usetikzlibrary{shapes.geometric, arrows, positioning}
\usetikzlibrary{cd}
\usetikzlibrary{babel,chains,matrix,positioning,scopes,calc,decorations,
	decorations.pathreplacing,automata}
\usepackage[all]{xy}
\usepackage{color}
\usepackage{verbatim}
\usepackage[makeroom]{cancel} 
\usepackage{faktor}
\usepackage{subfiles} 
\usepackage{soul}
\usepackage{graphicx}

\usepackage{imakeidx}
\makeindex[title=Notes,columns=1]

\usepackage{mathrsfs}
\DeclareMathAlphabet{\mathpzc}{OT1}{pzc}{m}{it}

\DeclareMathOperator{\SV}{SV}
\DeclareMathOperator{\RND}{RND}

\newtheorem{theorem}{Theorem}[section]
\newtheorem*{theorem*}{Theorem}				    			
\newtheorem{observation}[theorem]{Observation}			                
\newtheorem{?}[theorem]{Question}	                        

\newtheorem{proposition}[theorem]{Proposition}				
\newtheorem{corollary}[theorem]{Corollary}					

\theoremstyle{remark} 

\newcommand{\set}[1]{\left\{#1\right\}}				        

\newcommand{\NN}[0]{\mathbb{N}}

\usepackage{xcolor, hyperref}
\definecolor{orcidlogocol}{HTML}{A6CE39}

\DeclareRobustCommand{\orcidicon}{%
	\begin{tikzpicture}
	\draw[lime, fill=lime] (0,0) 
	circle [radius=0.16] 
	node[white] {{\fontfamily{qag}\selectfont \tiny ID}};
	\draw[white, fill=white] (-0.0625,0.095) 
	circle [radius=0.007];
	\end{tikzpicture}
	\hspace{-2mm}
}

\newcommand{\orcid}[1]{\href{https://orcid.org/#1}{\textcolor[HTML]{A6CE39}{\orcidicon}}}

\begin{document}

\title{Generating Dominating Sets Using Locally-Defined Centrality Measures}

\author{Hunter~Rehm\orcid{0000-0002-1284-2032}, Robert~Kassouf-Short\orcid{0000-0003-2087-5528}, Puck~Rombach\orcid{0000-0002-8374-0797}}



\maketitle

\begin{abstract}
    The dominating set problem has many practical applications but is well-known to be NP-hard. Therefore, there is a need for efficient approximation algorithms, especially in applications such as ad hoc wireless networks. 
    Most distributed algorithms proposed in the literature assume that each node has knowledge of the network structure. We propose a distributed approximation algorithm that uses two rounds of communication, and where each node has only local information, both in terms of network structure and dominating set assignment. First, each node calculates a local centrality measure to determine whether it is part of the dominating set $D$. The second round guarantees $D$ is a dominating set by adding any non-dominated nodes. We compare several centrality measures and show that the Shapley value, introduced in game theory, is theoretically motivated and performs well in practice on several synthetic and real-world networks.
     \\
{\bf Keywords -- } Network, Dominating, Approximation, Shapley.\\
{\bf MSC --} 05C69, 05C85 .
\end{abstract}

\section{Introduction}\label{sec:intro}

A dominating set in a network is a subset of nodes so that every node is either in the dominating set or has a neighbor in the set. Dominating sets and their variations have many applications~\cite{levin2020combinatorial,liu2010survey,yu2013connected}. Often, the objective is to find a dominating set of minimum size. 

Deciding whether or not a graph has a dominating set of size $k$ is NP-complete~\cite{garey1979guide}.
The current fastest exact algorithm for identifying a minimum dominating set was introduced in 2011 by Van Rooij and Bodlaender~\cite{van2011exact} and runs in $O(1.4969^n)$ time. Certain classes of graphs admit polynomial-time approximation schemes for the minimum dominating set problem. One example of such a  class is unit disk graphs. This class includes random geometric graphs, as well as many types of ad hoc wireless networks~\cite{cheng2003polynomial,hunt1998nc,nieberg2008approximation}. Instead of identifying a dominating set globally, there are also distributed algorithms where decisions are made at the level of individual nodes with limited rounds of communications~\cite{barenboim2018fast,dai2004extended,ghaffari2018derandomizing,jia2002efficient,liang2000virtual,wu2001dominating}. In most of these distributed algorithms, it is assumed that each node has full global information regarding the structure of the network, but only local information regarding the dominating set assignment. In some applications, nodes may have none or limited global information. A pertinent example of such an application is ad-hoc and wireless sensor networks~\cite{garcia2007wireless,khan2016wireless,singh2018survey,yick2008wireless}. 

The Space Communication and Navigation (SCaN) program office manages NASA's space communication activities.
One of their goals is to understand satellite communication, including routing, radio technology, and network modeling.
In the future of space science, autonomous robotic missions to other planets will be a primary means for humans to explore our solar system.
In particular, satellite swarm missions akin to HelioSwarm have the opportunity to offer us great insight into the nature of other planets \cite{HelioSwarm-RAJAN2022554}.
The power of satellite swarm missions comes from the wide variety of different roles that small satellites can serve.
However, to enable interplanetary satellite swarm missions, it will be essential that satellites can autonomously organize and optimize the variety of roles available to the swarm.

One motivation for this paper comes from this consideration.
Two general roles that satellites fall into are communications satellites and data collection satellites.
Future swarms may have the flexibility to switch between these roles as the needs of the swarm demand.
To enable the most efficient distribution of communications and data collection satellites, each satellite needs to be able to detect and adapt to the needs of the network autonomously.
Communications satellites should form a dominating set over the swarm so that each data collector has the capacity to collate and transmit their data, but we also want to minimize the number of communications satellites so as to maximize the amount of data the swarm can collect.

Finding a minimum dominating set in these types of applications is often not feasible, because there may not be a central processor with access to the full network structure. Instead, the objective is to find an approximately minimum dominating set in a distributed manner where nodes can access only local information in the network.

Let $G(V,E)$ be a simple, undirected  graph on $n$ nodes, and let $\gamma(G)$ be the \emph{domination number} of $G$; the cardinality of a minimum dominating set. One of the most fundamental results in the theory of dominating sets is an upper bound in terms of the minimum degree $\delta(G)$ of $G$:
\begin{equation}\label{alonupper}
    \gamma(G) \leq \frac{1+\ln(1+\delta(G))}{1+\delta(G)}n. 
\end{equation}  
This result was proved independently by~\cite{alon2016probabilistic,arnautov1974estimation,lovasz1975ratio,payan1975nombre}. The proof given in~\cite{alon2016probabilistic} has become a standard example of the power of the probabilistic method. It gives a simple, two-step probabilistic algorithm that yields a dominating set whose cardinality is in expectation of the upper bound in Inequality~\ref{alonupper}. The algorithm works as follows. In the first round, each vertex is assigned to a set $X$ with a fixed probability $p$, independently of other vertices. In the second round, each vertex not dominated by $X$ is added to a set $Y$. The output is an approximately minimum dominating set $X \cup Y$. This algorithm is both fast, as it is linear in the size of the graph, and highly local, as decisions are made at the vertex level with communication only to immediate neighbors.

Wu and Li~\cite{wu2001dominating}, proposed a similar algorithm with a deterministic first step. It uses the local clustering coefficient to select a dominating set~\cite{watts1998collective}.
This algorithm assigns a node $v$ to the dominating set if it has a clustering coefficient less than $1$. 
In this case, a connected dominating set is guaranteed, and there is no need for a second round.

Inspired by this style of algorithm, we propose an algorithm for generating dominating sets with the help of local centrality measures.
Our algorithm is not guaranteed to generate a minimum dominating set, but we show that it performs well in practice in finding small dominating sets.
In the first step, the algorithm selects vertices based on a centrality measure meeting a certain threshold. In the second step, more vertices are added as needed to guarantee a dominating set.

We compare several different centrality measures on synthetic and real-world networks. 
The best-performing measure in our analysis is the Shapley value, introduced in game theory~\cite{shapley1971cores}.
It is used in various contexts to measure the importance of an actor in game theory~\cite{ bozzo2016theory,bozzo2015vulnerability,michalak2013efficient,suri2008determining}.
There are various definitions of the Shapley value in the literature. Here we use the definition from~\cite{bozzo2015vulnerability}, which defines the Shapley value of a node as the sum of the reciprocals of the degrees of its neighbors. (These authors then subtract 1 from each value, which we will not do.)

In Section \ref{sec:maincontent}, we define the two-step algorithm and a few centrality measures that appear to effectively select small dominating sets.
We compare their performance in the two-step algorithm on various synthetic and real-world networks in Section~\ref{sec:SimulatedData} and \ref{sec:Application}, respectively. Section~\ref{sec:oddities} provides counterexamples to a few conjectures about the two-step algorithm that we developed from the simulations. Lastly, in Section~\ref{sec:Conclusion}, we discuss a few natural future directions.

\section{Dominating set algorithms using thresholds}\label{sec:maincontent}

In Section \ref{sec:CoreAlgorithm} we outline the two-step algorithm explicitly.
In Section \ref{sec:CentralityMeaures}, we define the centrality measures considered in this study with a few basic observations that motivate their use in this context.

\subsection{The two-step algorithm}\label{sec:CoreAlgorithm}
Let $c(v)$ be a locally computed node centrality measure on the node set of a network $V(G)$, and $\tau\geq0$ a threshold constant. 
The two-step algorithm has two rounds: the first round selects all nodes $v$ such that $c(v)>\tau$, and the second round selects all nodes $v$ such that neither $v$ nor any of its neighbors were selected in the first round.
The second step ensures that the final set of selected nodes forms a dominating set in $G$. Let $N(v) = \{u\in V|(u,v)\in E(G)\}$.

\begin{center}
\noindent\textbf{Two-step Algorithm}
\end{center}
\begin{enumerate}
    \item Let $X$ be the set of nodes $v$ with $c(v)>\tau$. 
    \item Let $Y$ be the set of nodes $v$ such that $(N(v)\cup \{v\})\cap X=\emptyset$.
    \item Let $D=X \cup Y$.
\end{enumerate}

This algorithm can be executed in a distributed manner. 
To perform the first step, each node sends the necessary information to its neighbors so that each of them can calculate $c(v)$. 
Then, each node communicates to its neighbors whether they are in the set and carries out the second step. 
All of the measures we consider in Section \ref{sec:CentralityMeaures} can be computed using one round of communication with the immediate neighborhood of each node, supposing each node knows its degree. However, local or not, any centrality measure can be used in the algorithm. If the algorithm is implemented in a distributed manner and a local centrality measure is used, it runs in time $O(\Delta (G))$, where $\Delta(G)$ is the maximum degree of the network. In the application of ad-hoc wireless networks, this often implies sublinear or even constant run time.

This algorithm performs best when $X$ is close to a dominating set so that $Y$ is small.
Thus, the goal of the measure $c(v)$ is to capture nodes in $X$ likely to appear in a minimum dominating set.
In Section \ref{sec:CentralityMeaures}, we define a few centrality measures of interest.
We compare their performance in the two-step algorithm on various synthetic and real-world networks in Section~\ref{sec:SimulatedData} and \ref{sec:Application}, respectively.

\subsection{Centrality measures}\label{sec:CentralityMeaures}

Since we want our two-step algorithm to be local, we focus on centrality measures that can be determined locally. Specifically, we consider centrality measures that depend only on the immediate neighborhood of a node. 

{\bf Uniform random measure (URM):}
Alon investigates which threshold minimizes the expected size of the set constructed in the two-step algorithm and gives an upper bound on the domination number~\cite{alon1990transversal}. Refer to Section~\ref{sec:intro} for more information.

{\bf Clustering coefficient (ICC):}
In~\cite{wu2001dominating}, Li and Wu show that the collection of all nodes with two non-adjacent neighbors is a connected dominating set.
This is quantified by finding the proportion of non-adjacent neighbors, known as the \textit{clustering coefficient}~\cite{watts1998collective} of a node $v$
$$CC(v) = \frac{t_v}{\binom{k_v}{2}} = \frac{2t_v}{k_v(k_v-1)}.$$
Here, $t_v$ is the number of triangles that node $v$ is a part of, and $k_v$ is the degree of node $v$.

Based on simulated data, we suspect that the nodes representative of stronger dominating set candidates tend to have lower values of $CC(v)$.
For this reason, we let 
$$ICC(v) = 1-CC(v) = \frac{\binom{k_v}{2} - t_v}{\binom{k_v}{2}}.$$
It is worth noting that $ICC$ is not a new centrality measure but rather our adjustment of the clustering coefficient that allows for more convenient comparisons with the other centrality measures, as they all tend to correlate positively with the likelihood of appearing in minimum dominating sets.

Using the language of our paper, Li and Wu show that the first step of two-step algorithms constructs a connected dominating set with measure $ICC$ and a threshold of $0$.

\textbf{Relative neighbor degree ($\RND$):} Ai, Li, Su, Jiang, and Xiong define the neighbor-degree centrality of $v$ as $ND(v) = \frac{\sum_{u\in N(v)}k_u}{k_i}$~\cite{ai2016node}.
We propose the \textit{relative neighbor-degree centrality} of $v$ as $$\RND(v) = \frac{k_v}{ND(v)} = \frac{k_v^2}{\sum_{u\in N(v)}k_u}$$
which measures how the average degree of nodes in $N(v)$ compares to the degree of $v$ itself. 
If a node $v$ has a relative neighbor degree of $1$, then the average degree of the neighbors of $v$ equals the degree of $v$. 
If $v$ has a relative neighbor degree larger than $1$, then $v$ makes for a good dominating set candidate as its degree is larger than its neighbor's degree on average.

\textbf{Shapley value ($\SV$):} The Shapley value was originally proposed in \cite{shapley1971cores}, but was not introduce in network theory until later. Bozzo, Franceschet, and Rinaldi describe the Shapley value of a node $v$ in a network as the sum of the reciprocals of the neighbors' degrees, 
$$SV(v) = \sum_{u\in N(v)}\frac{1}{deg(u)}\cite{bozzo2015vulnerability,shapley1971cores}.$$
A node with a high Shapley value is likely adjacent to many low-degree nodes.

For the remainder of the paper, we compare all four measures and specifically focus on how $\SV$ and $\RND$, as they seem to perform best in this setting.
From some simple analysis, these two measures are related.
For instance, Proposition \ref{prop:IDbigger} shows that the number of nodes selected in the first round of the two-step algorithm with $\SV$ is always greater than that with $\RND$.

\begin{proposition}\label{prop:IDbigger}
For node $v$ in a network $N$, if $\RND(v)\geq \tau$, then $\SV(v)\geq \tau$.
\end{proposition}

\begin{proof}
Suppose that $\RND(v)\geq \tau$.
Then the average degree of a node in $N_v$ is less than $k_v/\tau$.
By Jensen's inequality, the average of the set $\{1/k_u|u\in N_v\}$ is at least $\tau/k_v$.
Hence, $\SV(v) = \sum_{u \in N_v} 1/k_u\geq k_v\frac{\tau}{k_v} = \tau$.
\end{proof}

We give the expected value of $\SV$ in Proposition \ref{prop:average} below.
This proposition shows that even though $\SV$ is a locally computed measure, we know how each node will compare to the average without knowing anything about the network's structure. The expected value of $\RND$ is not as easily described.

\begin{proposition}\label{prop:average}
Let $G = (V,E)$ graph.
The expected Shapley value of a node $v\in V$ is $1$.
\end{proposition}

\begin{proof}
Let $A$ be the adjacency matrix of a network $N = (V,E)$ with $n = |V|$.
Let $M$ be the stochastic matrix with
$M_{ij} = \frac{A_{ij}}{k_j}$
and $\mathbf{1}$ be the $n$-dimensional vector with $1$ in each entry.
Notice that $(M\cdot \mathbf{1})_i = \SV_i$ for each $i\in V$.
Since the stochastic matrix preserves the average of vectors, then the expected Shapley value is $1$.
\end{proof}

From Proposition \ref{prop:average}, $\tau=1$ is a natural choice of threshold if we use $\SV$ in our two-step algorithm.
To round out this section, we give additional evidence suggesting that $\tau=1$ is a reasonable choice of threshold in a vacuum.

\begin{observation}\label{obs:degreeone}
If $\deg(v)=1$, then $v$ is adjacent to a node $u$ with $\SV(u)\geq 1$.
\end{observation}

\begin{observation}\label{obs:RNDgeqone}
    If $\deg(v) \geq \deg(u)$ for all $u \in N(v)$, then $\RND(v)\geq 1$.
\end{observation}

According to Observation \ref{obs:degreeone}, a threshold of $\tau = 1$ will, among others, select all nodes adjacent to degree $1$ nodes in the first round of the algorithm.
We will see in Section \ref{sec:SimulatedData} and \ref{sec:oddities}, other thresholds better find minimal dominating sets using the two-step algorithm. 

\section{Simulated data}\label{sec:SimulatedData}

We test our algorithm on simulated graphs generated from three random graph families. 
We generate $100$ networks for each family with $100$ nodes and an expected degree of $10$.
Below we describe the families and present the computational results.

\textbf{Erd\H{o}s-R\'enyi Model}~\cite{erdos1959random}
In Figure \ref{fig:erdos}, we compare the average size of a dominating set created by the two-step algorithm using a variety of thresholds over instances of Erd\H{o}s-R\'enyi graphs.
We sample each graph using 100 nodes and a probability $p = 0.1$. 

\textbf{Random Geometric Graphs}~\cite{clark1990unit}
In Figure \ref{fig:geo}, we compare the size of the dominating set found by the two-step algorithm with different threshold values in random geometric graphs on the unit square with $100$ nodes and a radius $r= 0.178$. The motivation for using this radius is that the average degree $100(0.178)^2$ is about $10$, consistent across the other random graph families.

\textbf{Chung-Lu Model}~\cite{aiello2001random,chung2002average,chung2002connected}
We use the Chung-Lu model to sample graphs with a longer-tailed degree distribution and create the degree sequence using a negative binomial degree distribution.
Figure \ref{fig:chung} compares the size of the dominating set for each threshold on 100 different graphs sampled using the Chung-Lu model.
We create the degree sequence using the negative binomial distribution with $n = 1$ and $p = 0.1$. 

In each of Figures \ref{fig:erdos}, \ref{fig:geo}, and \ref{fig:chung}, $\SV$ and $\RND$ construct the smallest dominating sets across all thresholds, with $\SV$ slightly outperforming $\RND$.
Notice that the minimal dominating set in each family typically occurs at a threshold greater than 1 rather than precisely at 1.
Because $\SV$ generally performed best, we focus on the Shapley value in Section \ref{sec:Application}.

\begin{figure}[htbp]
    \centering
    \includegraphics[scale = .48]{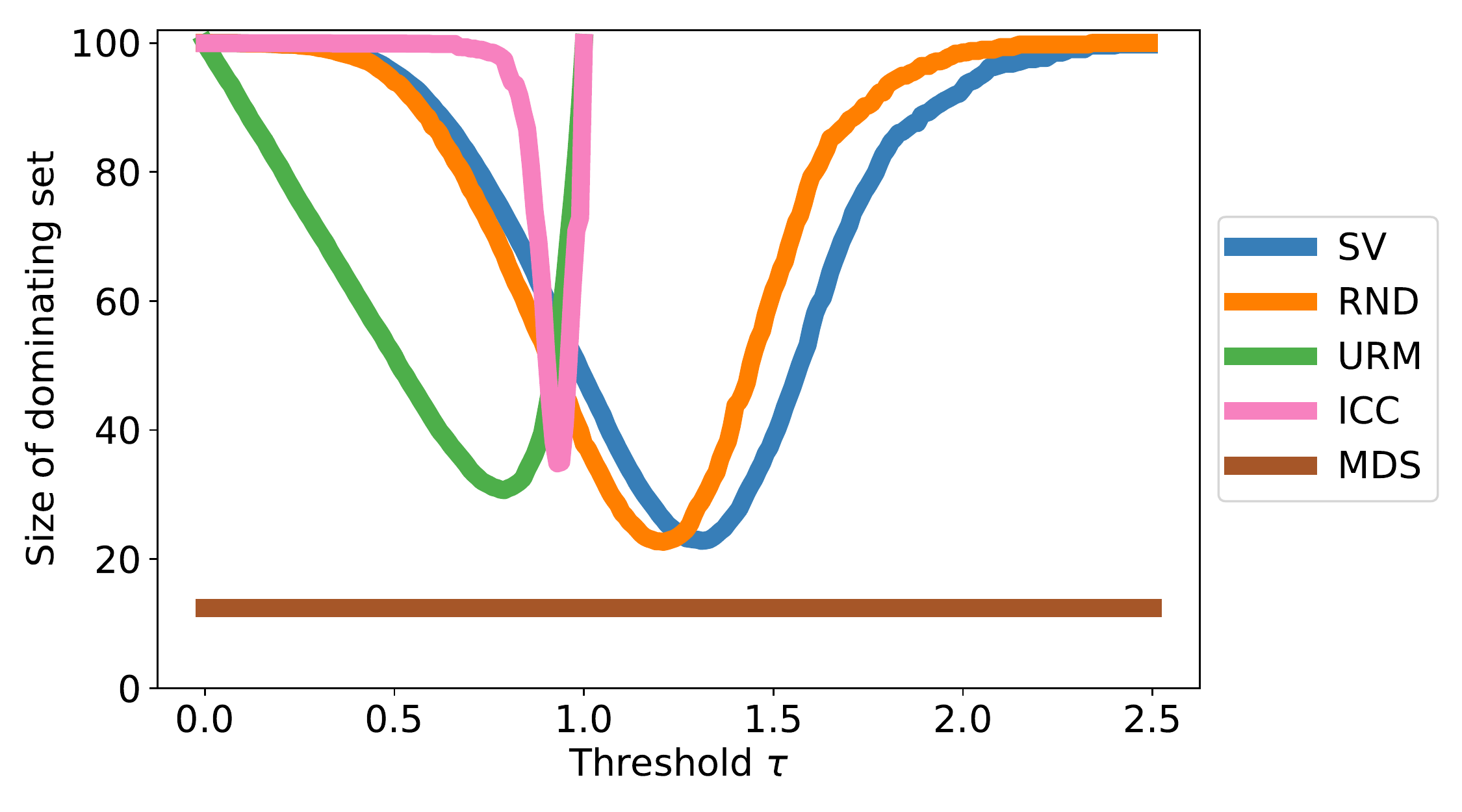}
    \caption{Erd\H{o}s-R\'enyi Model with 100 nodes and probability $0.1$.}
    \label{fig:erdos}
\end{figure}
\begin{figure}[htbp]
    \centering
    \includegraphics[scale = .48]{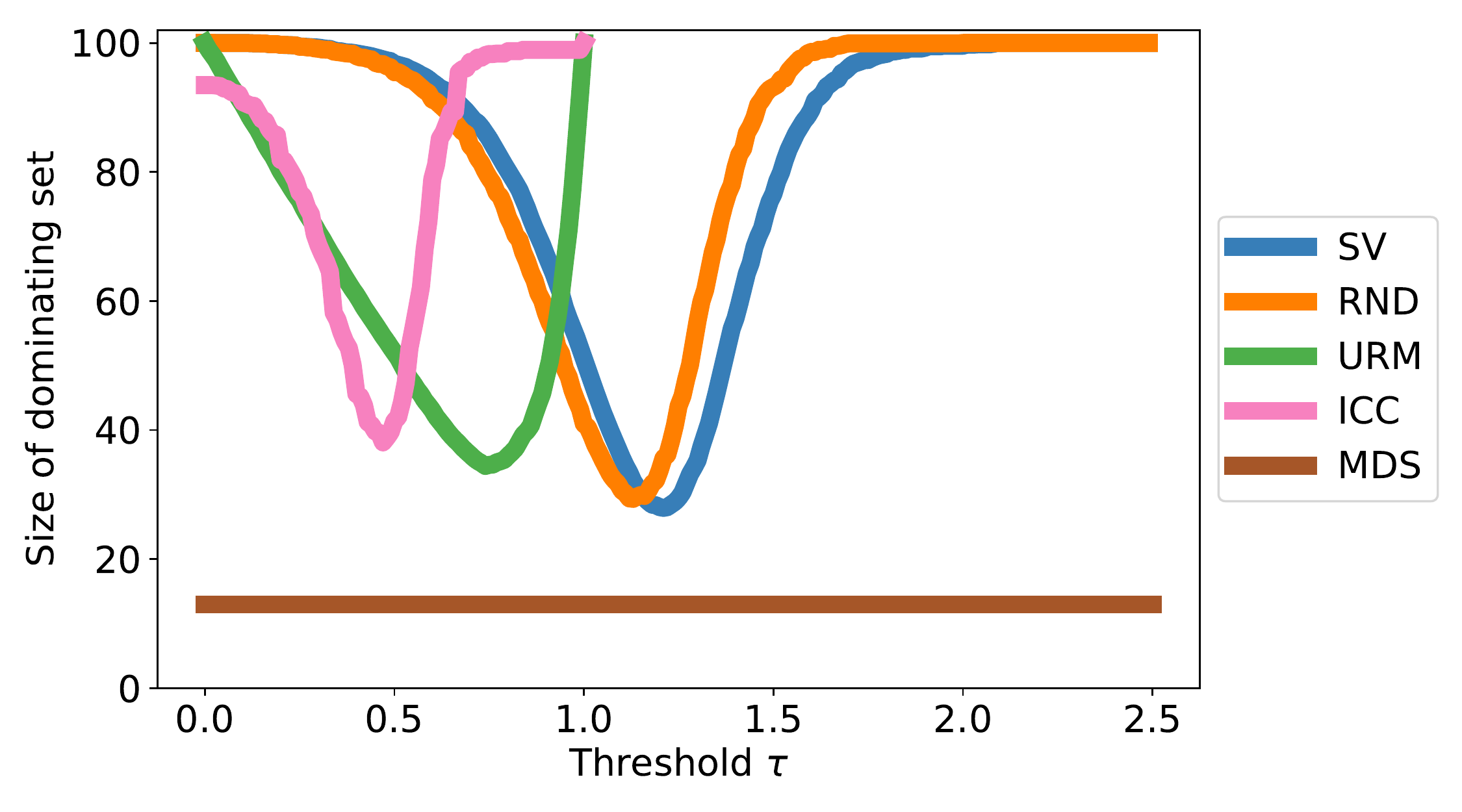}
    \caption{A comparison of the centrality measures in the local domination algorithm for random geometric graphs $G(100,0.178)$.}
    \label{fig:geo}
\end{figure}
\begin{figure}[htbp]
    \centering
    \includegraphics[scale = .48]{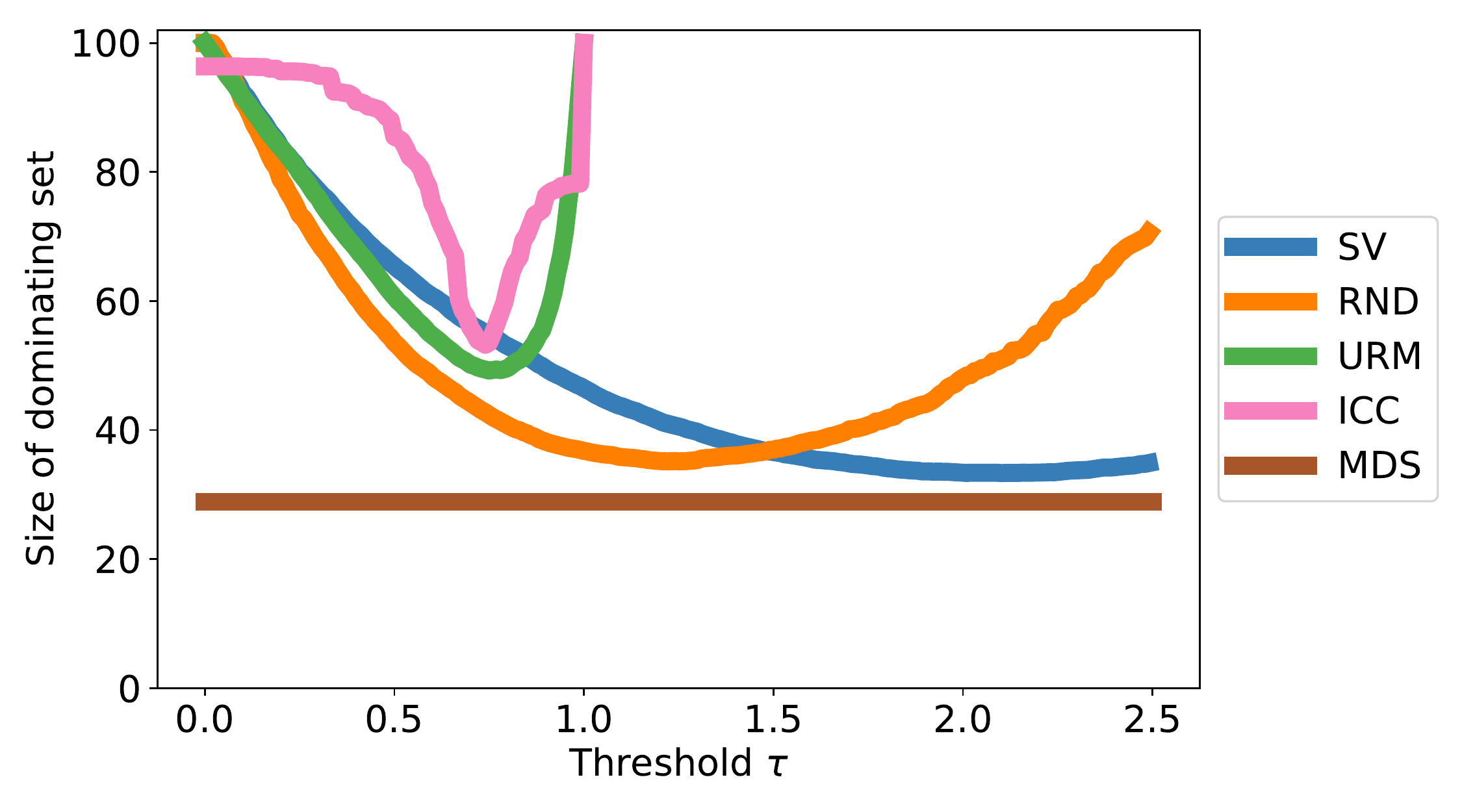}
    \caption{Chung-Lu model with 100 nodes and a degree distribution sampled from a negative binomial distribution.}
    \label{fig:chung}
\end{figure}

\section{Real-world data}\label{sec:Application}

We test the two-step algorithm using the Shapley value on two real-world networks: the Iridium satellite network and the European Natural Gas Pipeline network.
We first show how the algorithm performs with $\tau = 1$.
In the case of the Iridium network, we see a case where choosing a different threshold offers us a smaller dominating set.

\subsection{Iridium Satellite Network}
The Iridium satellites form a network in Low-Earth Orbit.
Figure \ref{fig:soap1} indicates which nodes are selected in the two-step algorithm (in pink) with the Shapley value and threshold $\tau = 1$.
Figure \ref{fig:soap12} shows the dominating set chosen by the two-step algorithm with the Shapley value when $\tau = 1.2$.

\begin{figure}[htbp]
    \begin{subfigure}{.49\textwidth}
        \centering
        \includegraphics[angle = 0,scale = .23]{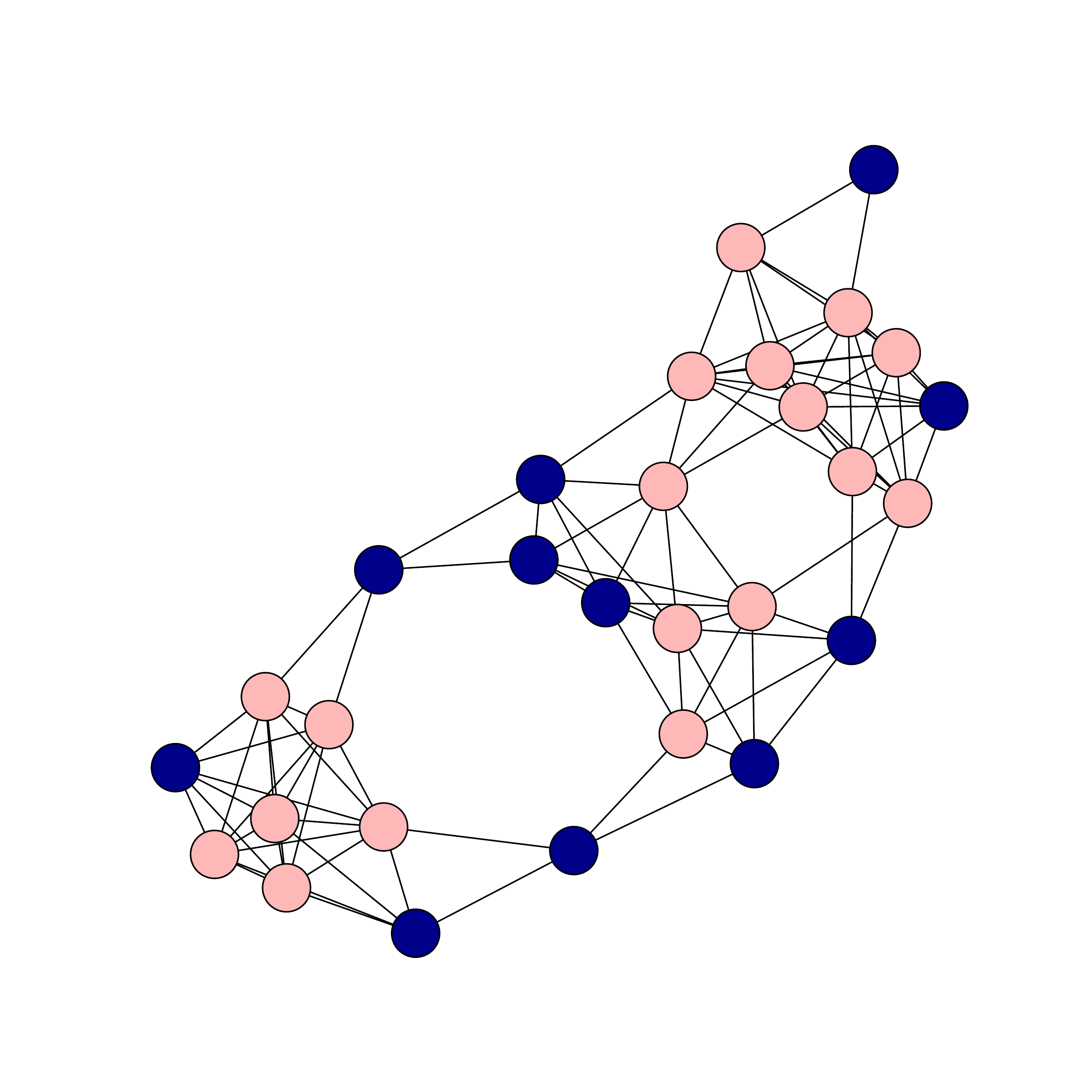}
        \caption{$\tau = 1$}
        \label{fig:soap1}
    \end{subfigure}
    \hfill
    \begin{subfigure}{.49\textwidth}
        \centering
        \includegraphics[angle = 0,scale = .23]{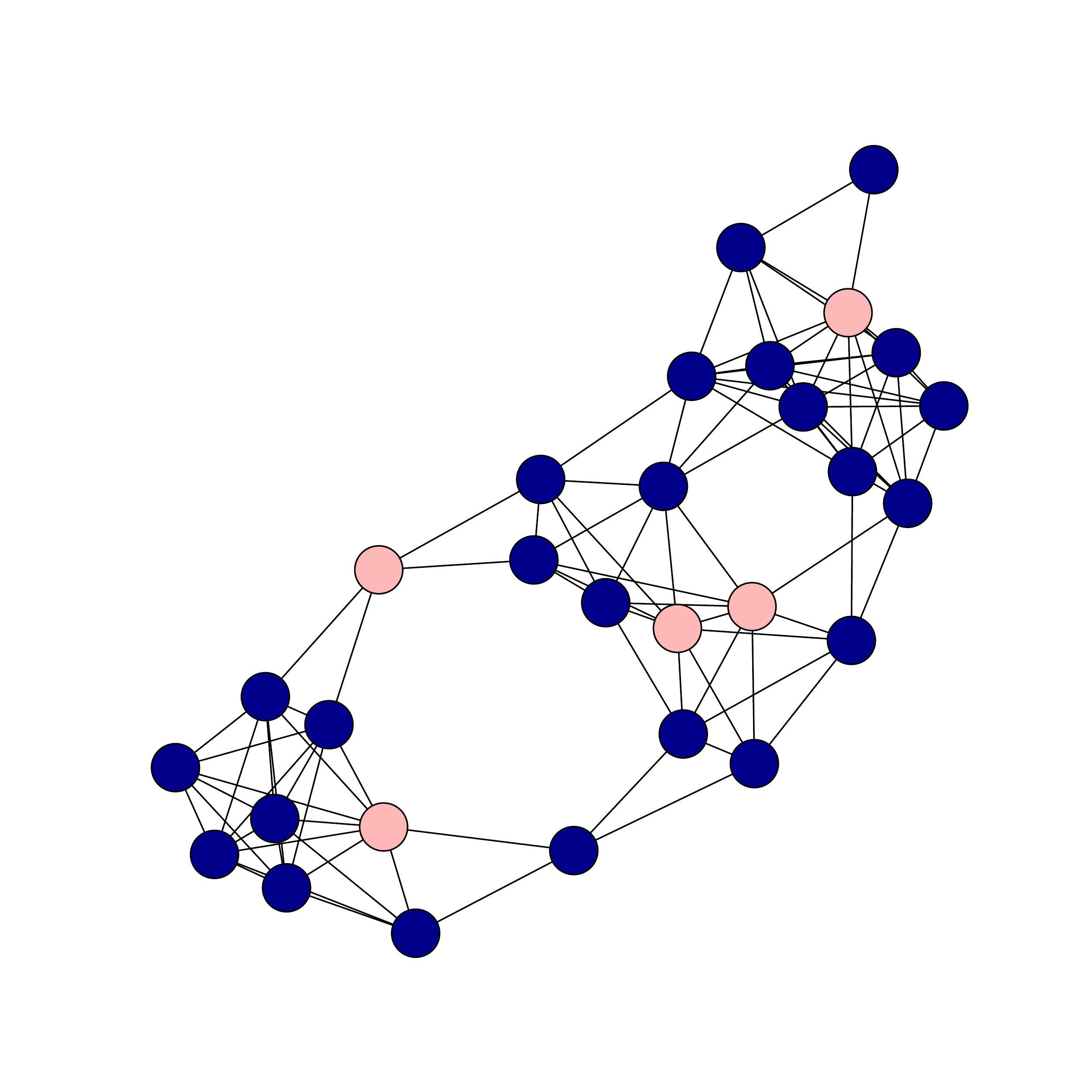}
        \caption{$\tau = 1.2$}
        \label{fig:soap12}
    \end{subfigure}
    \caption{Comparing thresholds in the two-step algorithm with $\SV$ centrality measure on an instance of the Iridium satellite network. The color pink indicates a node in the dominating set.}
\end{figure}

\subsection{European Natural Gas Pipeline Network}

The International Energy Agency (IEA) data, analysis, and policy recommendations have collected data from 31 participating countries to create the European natural gas network (ENGN).
We use the Shapley value in the two-step algorithm to find a dominating set. 
In Figure \ref{fig:piplines}, the nodes colored pink are selected by the two-step algorithm with $\SV$ centrality measure and a threshold of $\tau = 1$.  

Figure \ref{fig:IEA_subgraph} is a subgraph of the ENGN where the Shapley value scales each node. 
This visual represents how the Shapley value can vary quite a bit from node to node.

\begin{figure}[H]
    \centering
    \includegraphics[scale = .4]{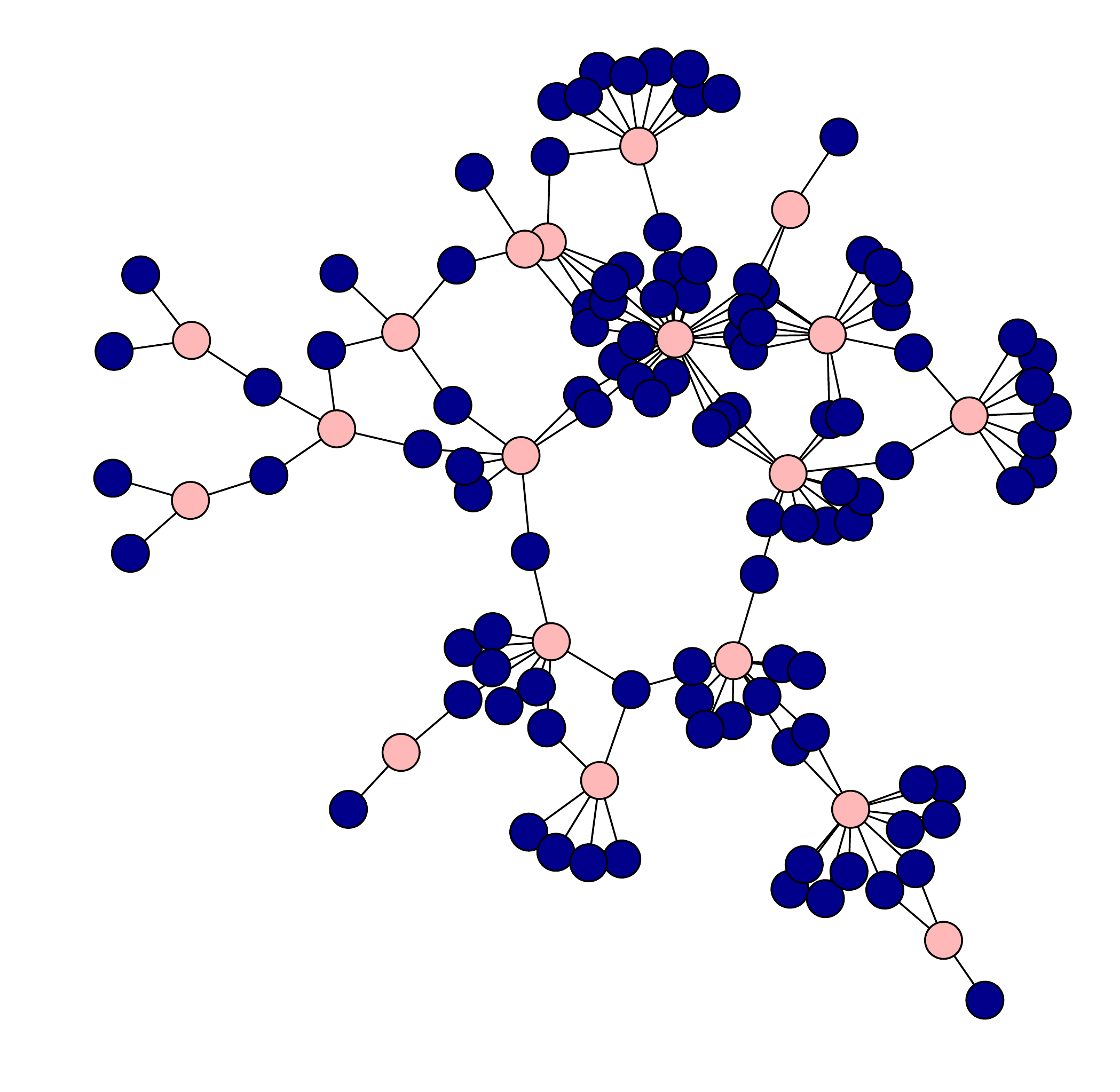}
    \caption{European Natural Gas Pipeline Network with dominating set.}
    \label{fig:piplines}
\end{figure}

\begin{figure}[H]
    \centering
    \includegraphics[scale = .18]{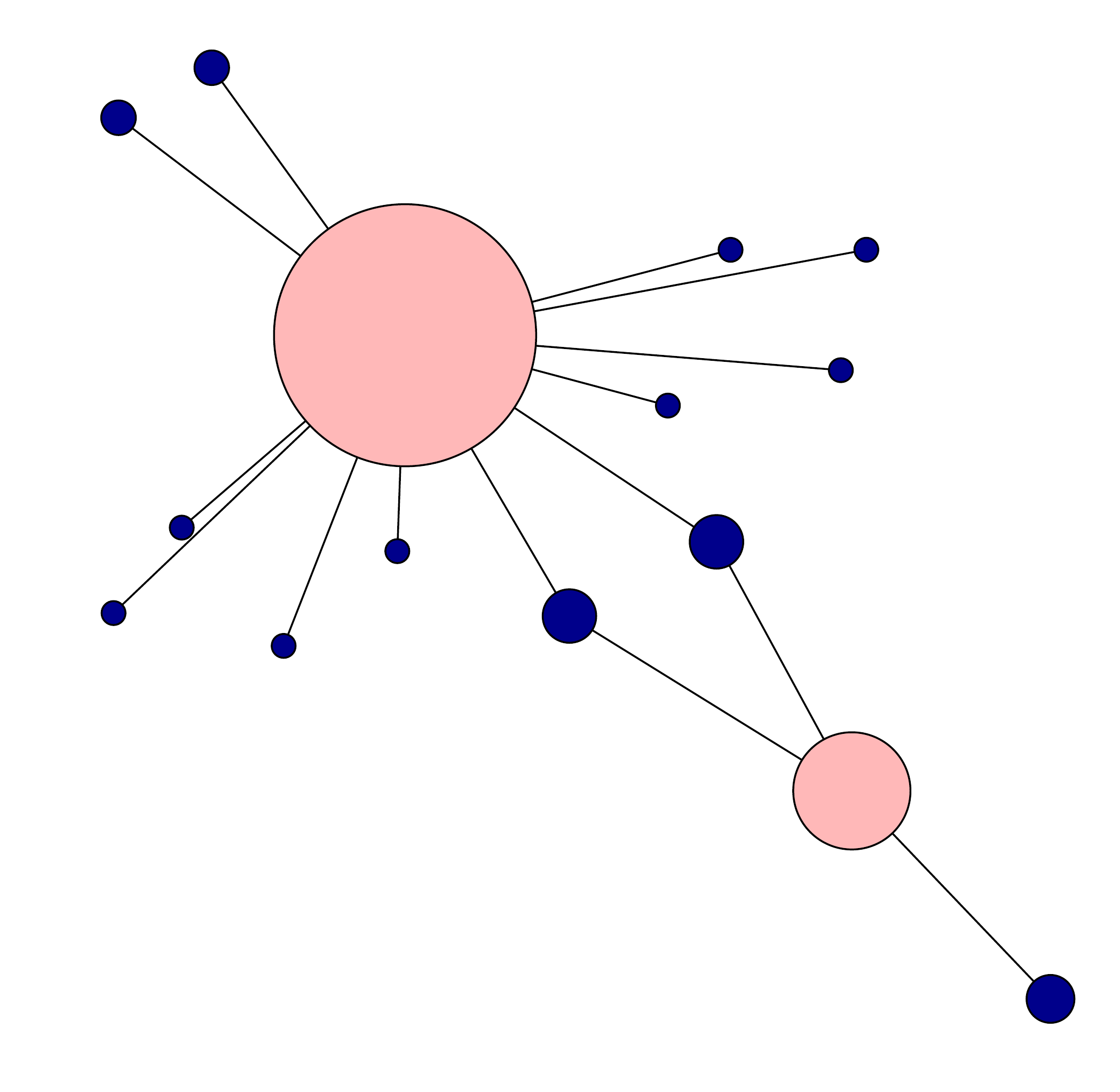}
    \caption{Subgraph of the European Natural Gas Pipeline Network. Each node size is scaled to its Shapley value.}
    \label{fig:IEA_subgraph}
\end{figure}

\section{Oddities}\label{sec:oddities}

This section presents a set of counterexamples to natural conjectures regarding the two-step algorithm with the $\SV$ centrality measure. These counterexamples are enlightening, revealing exciting relationships between the measures we examined here and some general oddities surrounding the two-step algorithm.

With a threshold in the first filtering for the two-step algorithm, much of our questioning centered on choosing the best threshold.
Inspired by Wu and Li \cite{wu2001dominating}, we initially sought a universal threshold for constructing a dominating set.
If there were such a threshold for $\SV$ or $\RND$, we could reduce our two-step algorithm to just one step: compute the measure and determine whether or not you are in the dominating set.
If the threshold is $ 0 $, only one step is necessary, but perhaps a small non-zero threshold will also guarantee a dominating set in one step.

It turns out that there is no such non-zero threshold, but our construction to show this is unwieldy and unlikely to appear in graphs with relatively small numbers of nodes.
We present this construction in the proof of the following proposition.

\begin{proposition}
\label{prop:NoUniversal}
Given a threshold $\tau>0$, let $S_{\tau}$ be the set of nodes such that $\SV(v)>\tau$.
There exists a graph $G$ such that $S_{\tau}$ is not a dominating set of $G$.
\end{proposition}

\begin{proof}
We construct a graph with node $a$ that satisfies the following properties:
\begin{enumerate}
    \item For node $a$, $\SV(a) < \tau$.
    \item For all nodes $b \in N(a)$, $\SV(b)<\tau$.
\end{enumerate}

Given $\tau>0$, there exists some $N_1 \in \NN$ such that $\frac{1}{N_1}<\tau$.
So, we can construct a set of $N_1$ nodes $\set{b_i}_{i=1}^{N_1}$ such that each $b_i$ is adjacent to $a$.
Thus, $\deg(a)=N_1$, and $N(a) = \set{b_i}_{i=1}^{N_1}$.

Since $N_1,\tau>0$, there exists $N_2 \in \NN$ such that $N_2>\frac{N_1}{\tau}$.
For each $b_i$, construct $N_2-1$ nodes $\set{c_{i,j}}_{j=1}^{N_2-1}$ such that each is adjacent to $b_i$.
Thus, $\deg(b_i) = N_2$, and $N(b_i) = \set{a}\cup \set{c_{i,j}}_{j=1}^{N_2-1}$ for each $i$.
Then, by construction, $\SV(a)= \frac{N_1}{N_2}<\tau$.

Since $\frac{1}{N_1}<\tau$, $\tau-\frac{1}{N_1}>0$, and since $N_2 \in \NN$, $N_2-1\geq 0$.
Thus, there exists $N_3 \in \NN$ such that $N_3 > \frac{N_2-1}{\tau-\frac{1}{N_1}}$.
So, for each $c_{i,j}$, construct $N_3-1$ nodes $\set{d_{i,j,k}}_{k=1}^{N_3-1}$ such that each is adjacent to $c_{i,j}$.
Thus, $\deg(c_{i,j}) = N_3$ for each $j$.
So, by construction, for each $b_i$,
\[
\SV(b_i) = \frac{1}{\deg(a)} + \sum\limits_{j=1}^{N_2-1}\frac{1}{\deg(c_{i,j})} = \frac{1}{N_1} + \frac{N_2-1}{N_3} < \tau.
\]
Thus, any graph containing this construction as a subgraph would satisfy the conditions above.
\end{proof}

From the construction in Proposition \ref{prop:NoUniversal}, and the relationship between $\RND$ and $\SV$ given in Proposition \ref{prop:IDbigger}, we get the following result.

\begin{corollary}
\label{cor:NoUniversalRND}
Given a threshold $\tau>0$, let $R_{\tau}$ be the set of nodes such that $\RND(v)>\tau$.
There exists a graph $G$ such that $R_{\tau}$ is not a dominating set of $G$.
\end{corollary}

The choice of threshold then depends on the graph structure we are examining.
However, in the plots we generated, the graphs appeared to have a consistent shape, including a single local minimum.
For the minimum to be unique, one condition would be that for any three thresholds $\tau < \sigma < \rho$, one could never see $f(\tau)<f(\sigma)>f(\rho)$.

Unfortunately, this is not true in general.
Figure \ref{fig:monotone} depicts one counterexample of this idea.
In particular, when we use $\SV$ with the two-step algorithm at three different thresholds ($5/6$, $1$, and $11/6$), the size of the dominating set generated goes from $19$ to $20$ and back to $19$.

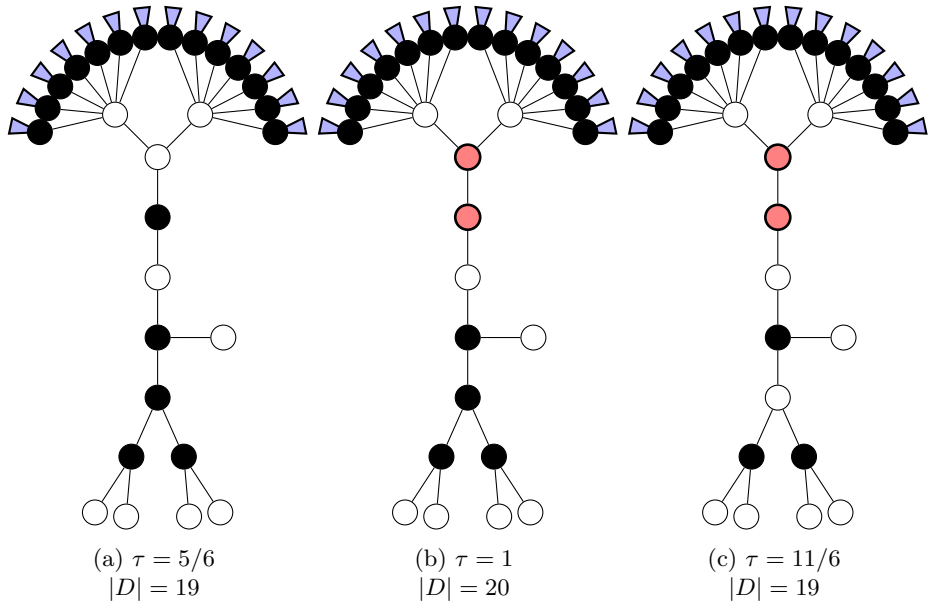
\begin{figure}[htbp]
    \begin{subfigure}{.32\textwidth}
        \centering
        
\begin{tikzpicture}[scale = .8]
    \node[circle,draw] (a) at ({45}:1){};
    \node[circle,draw] (b) at ({135}:1){};
    \foreach \a in {1,...,7}{
        \draw[fill=blue!30,line width=.8pt] (\a*12:2) -- (\a*12-3:2.5) -- (\a*12+3:2.5) -- (\a*12:2);
        \node[circle,draw,fill, black] (u\a) at ({\a*12}:2){};
        \draw (a) -- (u\a);
    }
    \foreach \a in {8,...,14}{
        \draw[fill=blue!30,line width=.8pt] (\a*12:2) -- (\a*12-3:2.5) -- (\a*12+3:2.5) -- (\a*12:2);
        \node[circle,draw,fill,black] (u\a) at ({\a*12}:2){};
        \draw (b) -- (u\a);
    }

    \node[circle,draw] (c) at (0:0){};
    \draw (a) -- (c);
    \draw (b) -- (c);
    \node[circle,draw,fill,black] (d) at (270:1){};
    \draw (d) -- (c);
    \node[circle,draw] (e) at (270:2){};
    \draw (d) -- (e);
    \node[circle,draw,fill,black] (f) at (270:3){};
    \node[circle,draw] (f1) at (290:3.193){};
    \draw (e) -- (f);
    \draw (f1) -- (f);
    \node[circle,draw,fill,black] (g) at (270:4){};
    \draw (f) -- (g);
    \node[circle,draw,fill,black] (h1) at (265:5){};
    \draw (g) -- (h1);
    \node[circle,draw,fill,black] (h2) at (275:5){};
    \draw (g) -- (h2);
    \node[circle,draw] (h11) at (260:6){};
    \draw (h1) -- (h11);
    \node[circle,draw] (h12) at (265:6){};
    \draw (h1) -- (h12);
    \node[circle,draw] (h21) at (275:6){};
    \draw (h2) -- (h21);
    \node[circle,draw] (h22) at (280:6){};
    \draw (h2) -- (h22);
\end{tikzpicture}
        \caption{$\tau = 5/6$\\$|D| = 19$}
        \label{fig:monotone56}
    \end{subfigure}
    \hfill
    \begin{subfigure}{.32\textwidth}
        \centering
        
\begin{tikzpicture}[scale = .8]
    \node[circle,draw] (a) at ({45}:1){};
    \node[circle,draw] (b) at ({135}:1){};
    \foreach \a in {1,...,7}{
        \draw[fill=blue!30,line width=.8pt] (\a*12:2) -- (\a*12-3:2.5) -- (\a*12+3:2.5) -- (\a*12:2);
        \node[circle,draw,fill, black] (u\a) at ({\a*12}:2){};
        \draw (a) -- (u\a);
    }
    \foreach \a in {8,...,14}{
        \draw[fill=blue!30,line width=.8pt] (\a*12:2) -- (\a*12-3:2.5) -- (\a*12+3:2.5) -- (\a*12:2);
        \node[circle,draw,fill,black] (u\a) at ({\a*12}:2){};
        \draw (b) -- (u\a);
    }

    \node[circle,fill=red!50,draw=black,line width=1pt] (c) at (0:0){};
    \draw (a) -- (c);
    \draw (b) -- (c);
    \node[circle,fill=red!50,draw=black,line width=1pt] (d) at (270:1){};
    \draw (d) -- (c);
    \node[circle,draw] (e) at (270:2){};
    \draw (d) -- (e);
    \node[circle,draw,fill,black] (f) at (270:3){};
    \node[circle,draw] (f1) at (290:3.193){};
    \draw (e) -- (f);
    \draw (f1) -- (f);
    \node[circle,draw,fill,black] (g) at (270:4){};
    \draw (f) -- (g);
    \node[circle,draw,fill,black] (h1) at (265:5){};
    \draw (g) -- (h1);
    \node[circle,draw,fill,black] (h2) at (275:5){};
    \draw (g) -- (h2);
    \node[circle,draw] (h11) at (260:6){};
    \draw (h1) -- (h11);
    \node[circle,draw] (h12) at (265:6){};
    \draw (h1) -- (h12);
    \node[circle,draw] (h21) at (275:6){};
    \draw (h2) -- (h21);
    \node[circle,draw] (h22) at (280:6){};
    \draw (h2) -- (h22);
\end{tikzpicture}
        \caption{$\tau = 1$\\$|D| = 20$}
        \label{fig:monotone1}
    \end{subfigure}
    \hfill
    \begin{subfigure}{.32\textwidth}
        \centering
        
\begin{tikzpicture}[scale = .8]
    \node[circle,draw] (a) at ({45}:1){};
    \node[circle,draw] (b) at ({135}:1){};
    \foreach \a in {1,...,7}{
        \draw[fill=blue!30,line width=.8pt] (\a*12:2) -- (\a*12-3:2.5) -- (\a*12+3:2.5) -- (\a*12:2);
        \node[circle,draw,fill, black] (u\a) at ({\a*12}:2){};
        \draw (a) -- (u\a);
    }
    \foreach \a in {8,...,14}{
        \draw[fill=blue!30,line width=.8pt] (\a*12:2) -- (\a*12-3:2.5) -- (\a*12+3:2.5) -- (\a*12:2);
        \node[circle,draw,fill,black] (u\a) at ({\a*12}:2){};
        \draw (b) -- (u\a);
    }

    \node[circle,fill=red!50,draw=black,line width=1pt] (c) at (0:0){};
    \draw (a) -- (c);
    \draw (b) -- (c);
    \node[circle,fill=red!50,draw=black,line width=1pt] (d) at (270:1){};
    \draw (d) -- (c);
    \node[circle,draw] (e) at (270:2){};
    \draw (d) -- (e);
    \node[circle,draw,fill,black] (f) at (270:3){};
    \node[circle,draw] (f1) at (290:3.193){};
    \draw (e) -- (f);
    \draw (f1) -- (f);
    \node[circle,draw] (g) at (270:4){};
    \draw (f) -- (g);
    \node[circle,draw,fill,black] (h1) at (265:5){};
    \draw (g) -- (h1);
    \node[circle,draw,fill,black] (h2) at (275:5){};
    \draw (g) -- (h2);
    \node[circle,draw] (h11) at (260:6){};
    \draw (h1) -- (h11);
    \node[circle,draw] (h12) at (265:6){};
    \draw (h1) -- (h12);
    \node[circle,draw] (h21) at (275:6){};
    \draw (h2) -- (h21);
    \node[circle,draw] (h22) at (280:6){};
    \draw (h2) -- (h22);
\end{tikzpicture}
        \caption{$\tau = 11/6$\\$|D| = 19$}
        \label{fig:monotone116}
    \end{subfigure}
    \caption{The size of the dominating set formed by the two-step algorithm with the Shapley value does not achieve a unique local minimum. The purple triangles represent 100 leaf nodes. A black node suggests that it was chosen in the first round of the two-step algorithm, and red suggests it was chosen in the second round.}
    \label{fig:monotone}
\end{figure}

We conclude our catalog of oddities by focusing on the threshold of 1.
Thanks to Proposition \ref{prop:average}, we know that $1$ holds a special place concerning the $\SV$ and $\RND$ values.

However, returning to the data represented in Figures \ref{fig:geo},\ref{fig:chung}, and \ref{fig:erdos}, there was one further counterexample to consider.
When we set the threshold at $1$, the dominating set generated with $\RND$ was generally larger than the dominating set developed with $\SV$. However, Figure \ref{fig:sizeatone} shows an example where the number of nodes selected in the two-step algorithm is smaller in $\SV$ (Figure~\ref{fig:SV}) than in $\RND$ (Figure~\ref{fig:RND}) using the threshold $\tau = 1$.
Here, the black nodes are chosen in the first round and the pink nodes in the second.

\begin{figure}[htbp]
     \centering
     \begin{subfigure}[b]{0.41\textwidth}
         \centering
          \begin{tikzpicture}
    \node[circle,draw,fill,black] (u0) at ({0}:0){};
    \foreach \a in {1,...,8}
    {
        \node[circle,draw,fill,black] (u\a) at ({\a*45}:1){};
    }
    \foreach \a in {9,...,24}
    {
        \node[circle,draw] (u\a) at ({(\a-8)*22.5+11.25}:2){};
    }
    \foreach \a in {1,...,8}
    {
        \draw (u0) -- (u\a);
    }
    \foreach \a in {1,...,8}
    {
    \pgfmathtruncatemacro{\b}{(2*\a+7};
    \pgfmathtruncatemacro{\c}{(2*\a+8};
        \draw (u\a) -- (u\b);
        \draw (u\a) -- (u\c);
    }
    \end{tikzpicture}
         \caption{$\SV$}
         \label{fig:SV}
     \end{subfigure}
     \hfill
     \begin{subfigure}[b]{0.41\textwidth}
         \centering
          \begin{tikzpicture}
    \node[circle,draw,fill,black] (u0) at ({0}:0){};
    \foreach \a in {1,...,8}
    {
        \node[circle,draw] (u\a) at ({\a*45}:1){};
    }
    \foreach \a in {9,...,24}
    {
        \node[circle,fill=red!50,draw=black,line width=1pt] (u\a) at ({(\a-8)*22.5+11.25}:2){};
    }
    \foreach \a in {1,...,8}
    {
        \draw (u0) -- (u\a);
    }
    \foreach \a in {1,...,8}
    {
    \pgfmathtruncatemacro{\b}{(2*\a+7};
    \pgfmathtruncatemacro{\c}{(2*\a+8};
        \draw (u\a) -- (u\b);
        \draw (u\a) -- (u\c);
    }
    \end{tikzpicture}
         \caption{$\RND$}
         \label{fig:RND}
     \end{subfigure}
        \caption{An example where the dominating set generated using $\RND$ is greater than the dominating set generated using $\SV$.}
        \label{fig:sizeatone}
\end{figure}
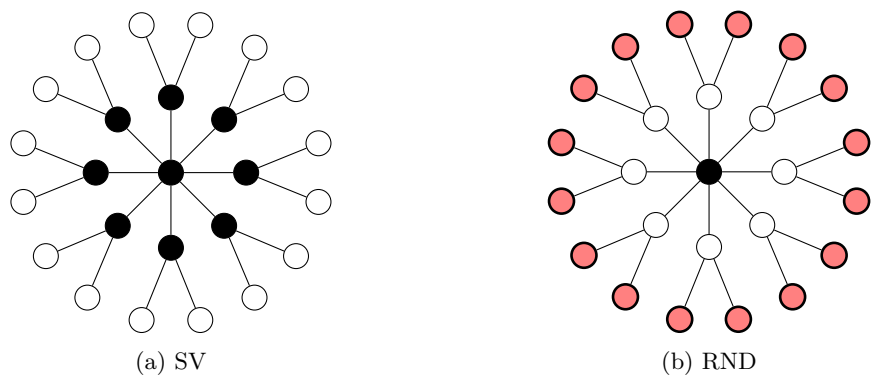

There is undoubtedly more exploration in characterizing the function described in the plots.
Both experimentally and theoretically, further investigation is warranted.

\section{Future directions}\label{sec:Conclusion}

Various algorithms aim to approximate a minimum dominating set for graphs \cite{lu2010survey}.
Our approach of using a two-step algorithm allows for a larger variety of node measures than we propose for generating dominating sets.
However, we focused on four local measures to highlight some methods that could be practical in robot swarm and autonomous system communications.

We see several directions to take this work into the future.
First, an analysis of different measures and their properties when plugged into the two-step algorithm may be an excellent method for comparing the effects of other measures.
Included with this, a more in-depth theoretical analysis of the two-step algorithm is warranted to explore its effectiveness further.

One concept we have toyed with is constructing bounds for the thresholds in various situations.
For example, one might wonder whether there are certain classes of networks that yield natural bounds.
For instance, with regular graphs, the Shapley value at each node is necessarily 1. In graphs that are regular or close to regular, it may be effective to combine for example the Shapley value with a random measure.
Predicting optimal thresholds based on network structure alone or time-varying network structure alone is a powerful tool in robotic mission design.
Moreover, if the ability to find bounds analytically proves intractable, passing this information into a neural network may be an exciting problem to pursue.
A Graph Neural Network that intakes graph structure and outputs a substantial threshold for achieving a dominating set could enable more optimizations, especially if this could be computed locally in some way.


\section*{Acknowledgements}\label{sec:Acknowledgements}
The authors would like to acknowledge the Space Communications and Navigation program at NASA for enabling helpful discussions and feedback throughout this work. 
Hunter Rehm is supported by the Vermont Space Grant Consortium (VTSGC) through the Graduate Fellowship Program.


\end{document}